\pdfoutput=1
\RequirePackage{ifpdf}
\ifpdf % We are running pdfTeX in pdf mode
\documentclass[pdftex]{sigma}
\else
\documentclass{sigma}
\fi

\numberwithin{equation}{section}

\newtheorem{Theorem}{Theorem}[section]
\newtheorem{Lemma}[Theorem]{Lemma}
\newtheorem{Proposition}[Theorem]{Proposition}
 { \theoremstyle{definition}
\newtheorem{Remark}[Theorem]{Remark} }

\newcommand{\BZ}{{\mathbb Z}}
\newcommand{\BC}{{\mathbb C}}
\newcommand{\sfrac}[2]{{#1}/{#2}}
\def\o#1{\overline{#1}}
\def\u#1{\underline{#1}}

\begin{document}

\allowdisplaybreaks

\newcommand{\arXivNumber}{1706.05155}

\renewcommand{\thefootnote}{}

\renewcommand{\PaperNumber}{069}

\FirstPageHeading

\ShortArticleName{An Elliptic Garnier System from Interpolation}

\ArticleName{An Elliptic Garnier System from Interpolation\footnote{This paper is a~contribution to the Special Issue on Elliptic Hypergeometric Functions and Their Applications. The full collection is available at \href{https://www.emis.de/journals/SIGMA/EHF2017.html}{https://www.emis.de/journals/SIGMA/EHF2017.html}}}

\Author{Yasuhiko YAMADA}

\AuthorNameForHeading{Y.~Yamada}

\Address{Department of Mathematics, Kobe University, Rokko, Kobe 657-8501, Japan}
\Email{\href{mailto:yamaday@math.kobe-u.ac.jp}{yamaday@math.kobe-u.ac.jp}}

\ArticleDates{Received June 20, 2017, in f\/inal form August 30, 2017; Published online September 02, 2017}

\Abstract{Considering a certain interpolation problem, we derive a series of elliptic dif\/fe\-ren\-ce isomonodromic systems together with their Lax forms. These systems give a multivariate extension of the elliptic Painlev\'e equation.}


\Classification{39A13; 33E05; 33E17; 41A05}

\renewcommand{\thefootnote}{\arabic{footnote}}
\setcounter{footnote}{0}

\section{Introduction}
There is a simple way to derive isomonodromic equations by studying suitable Pad\'e approximation or interpretation problem. It has been applied various examples both continuous and discrete (see \cite{NY16, Yamada09} and references therein). The aim of this paper is to apply this method to certain elliptic interpolation problems and derive a multivariate extension of the elliptic-dif\/ference\footnote{The $q$-dif\/ference limit of the obtained system is expected to be the one considered in \cite{OR16q}.} Painlev\'e equation \cite{ORG, Sakai}. This work is a natural generalization of \cite{NTY13}.

Recently, there have been some progress in multivariate elliptic isomonodromic systems. In~\cite{OR16, Rains11}, an elliptic analog of the Garnier system is constructed. In \cite{NY16}, an elliptic deformation of $q$-Garnier system is suggested from a geometric points of view. In \cite{ND}, certain elliptic analog of Garnier system is obtained from viewpoint of lattice equations. Moreover, a general framework of elliptic isomonodromic systems is established in~\cite{Rains16}. For the equations obtained in this paper, the proper isomonodromic interpretation and the relation to the constructions mentioned above are not clear so far. However, since the equations obtained in this paper are quite explicit, we expect that they will give a clue to elucidate the multivariate elliptic isomonodromic systems.

The paper is organized as follows. In Section~\ref{section2}, we set up our interpolation problem (\ref{eq:Pade-prob}): $\psi(z)\sim \frac{P(z)}{Q(z)}$. In Section~\ref{section3}, we derive two contiguous relations satisf\/ied by the interpolants $P(z)$ and $\psi(z)Q(z)$ (Theorem~\ref{th:cont}). These relations play the role of the Lax pair for the isomonodromic system. In Section~\ref{section4}, we analyze the Lax equations and derive the isomonodromic system as the necessary and suf\/f\/icient conditions for the compatibility (Theorem~\ref{th:44}). The proof becomes quite simple due to the use of the contiguous type Lax pair.

\section{Set up of the interpolation problem}\label{section2}
Fix $p, q \in \BC$ such that $|p|, |q|<1$. The elliptic Gamma function $\Gamma_{p,q}(z)$~\cite{Rui} and the theta function~$[z]$ (of base~$p$) are def\/ined as
\begin{gather*}
 \Gamma_{p,q}(z)=\frac{(pq/z;p,q)_\infty}{(z;p,q)_\infty}=\prod_{i,j\geq 0}\frac{1-z^{-1}p^{i+1}q^{j+1}}{1-z p^iq^j},\\
[z]=(z,p/z;p)_{\infty}=\prod_{i=0}^{\infty}\big(1-x p^i\big)\big(1-x^{-1}p^{i+1}\big).
\end{gather*}
They satisfy the following fundamental relations:
\begin{gather*}
\Gamma_{p,q}(q z)=[z]\Gamma_{p,q}(z),\qquad \Gamma_{p,q}(z/q)=[z/q]^{-1}\Gamma_{p,q}(z),\\
[p z]=-z^{-1}[z], \qquad [z]=-z [1/z],\qquad [z]=[p/z].
\end{gather*}
We also use the following notations:
\begin{gather*}
[z]_s=\frac{\Gamma_{p,q}(q^s z)}{\Gamma_{p,q}(z)}, \qquad[x_1,\ldots,x_l]_{s}=[x_1]_s \cdots [x_l]_s.
\end{gather*}
In particular, $ [z]_s=\prod\limits_{i=0}^{s-1}[q^iz]$ for $s \in \BZ_{\geq 0}$.

Fix $N \in \BZ_{\geq 2}$. Let $k, u_1, \ldots, u_{2N}$ be complex parameters satisfying a constraint
$\prod\limits_{i=1}^{2 N} u_i=k^N$, and def\/ine a function $\psi(z)$ as\footnote{Throughout the paper, any expression $a\cdots b/c\cdots d$ means the long fraction $\frac{a\cdots b}{c\cdots d}$.}
\begin{gather*}
\psi(z)=\prod_{i=1}^{2N}\frac{\Gamma_{p,q}(u_i/z)}{\Gamma_{p,q}(k/u_i z)}.
\end{gather*}
We also def\/ine a shift $T\colon x \mapsto \o{x}$ of parameters $x=k, u_i$ as
\begin{gather*}
\o{k}=k/q, \qquad \o{u_i}= \begin{cases}
u_i, & 1\leq i\leq N,\\
u_i/q, & N< i\leq 2N.
\end{cases}
\end{gather*}
This action is naturally extended to any functions $f=f(k, u_i)$ of parameters by $\o{f}=f(\o{k}, \o{u_i})$.

We put
\begin{gather*}%\label{eq:ell-MK}
\mu_1(z)=\frac{\psi(z/q)}{\psi(z)}= \prod_{i=1}^{2 N}\frac{[u_i/z]}{[k/u_iz]},\qquad
\mu_2(z)=\frac{\o{\psi}(z/q)}{\psi(z)}=\prod_{i=1}^N\frac{ [u_i/z]}{[k/u_{N+i}z]},\\
\mu_3(z)=\frac{\o{\psi}(z)}{\psi(z)}=\prod_{i=1}^N\frac{ [k/qu_iz]}{[u_{N+i}/qz]}.
\end{gather*}
These functions are $p$-periodic: $f(pz)=f(z)$, and satisfy
\begin{gather}\label{eq:MK-rel}
\mu_1(k/z)=\mu_1(z)^{-1}, \qquad \mu_3(k/qz)=\mu_2(z),
\end{gather}
due to the constraint $\prod\limits_{i=1}^{2 N} u_i=k^N$.

Let $f(z)$ be an elliptic function of degree $2d$ such that $p$-periodic and $h$-symmetric: $f(h/z)=f(z)$. Any such function can be written as $f(z)=\frac{\theta_{\rm num}(z)}{\theta_{\rm den}(z)}$, where $\theta_{\ast}(z)$ ($\ast={\rm num}, {\rm den})$ are $h$-sym\-metric entire function with common quasi periodicity: $\theta_{\ast}(p z)=(h/p z^2)^d \theta_{\ast}(z)$. The totality of such functions $f(z)$ form a linear space of dimension $d+1$.

For $m,n \in \BZ_{\geq 0}$, consider the interpolation problem\footnote{This is a kind of PPZ (prescribed poles and zeros) interpolation \cite{Zhe}.}
\begin{gather}\label{eq:Pade-prob}
\psi(q^{-s})=\frac{P(q^{-s})}{Q(q^{-s})}, \qquad s=0,1,\dots ,N=m+n,
\end{gather}
where \looseness=-1 $P(z)$ (resp.~$Q(z)$) are $k/q$-symmetric and $p$-periodic elliptic functions of order~$2m$ (resp.\ $2n$), with specif\/ied denominators $P_{\rm den}(z)$ (resp.~$Q_{\rm den}(z)$). For convenience, we will choose them~as
\begin{gather*}
P_{\rm den}(z)=[u_1/z,u_1q z/k]_m,\qquad Q_{\rm den}(z)=[k/u_2 z,q z/u_2]_n.
\end{gather*}

\section{Derivation of the contiguous relations}\label{section3}
Let $P(z)$, $Q(z)$ be solutions for the interpolation problem~(\ref{eq:Pade-prob}). We will compute the contiguous relations satisf\/ied by the functions $w(z)=P(z)$, $\psi(z)Q(z)$:
\begin{gather}
L_2\colon \ D_3(\sfrac{z}{q})w(z)-D_2(z)w(\sfrac{z}{q})+D_1(z)\o{w}(\sfrac{z}{q})=0,\nonumber\\
L_3\colon \ \o{D_1}(z)w(z)-D_2(z)\o{w}(z)+D_3(z)\o{w}(\sfrac{z}{q})=0.\label{eq:L23D}
\end{gather}
The coef\/f\/icients are determined by the Casorati determinants as
\begin{gather*}
D_1(z)=|{\bf u}(z), {\bf u}(z/q)|,\qquad D_2(z)=|{\bf u}(z), \o{\bf u}(z/q)|,\\
D_3(z)=|{\bf u}(z), \o{\bf u}(z)|, \qquad {\bf u}(z)=\left[\begin{matrix} P(z)\\\psi(z)Q(z)\end{matrix}\right].
\end{gather*}
Certain explicit formulas for $P(z)$, $Q(z)$ are available (see Remark~\ref{remark3.4}), however, we do not need them for the computations here.

\begin{Lemma} \label{lemma:X1} We have
\begin{gather*}%\label{eq:XF}
D_1(z)=\psi(z) \frac{[z,k/z]_{m+n} \big[k/z^2\big]}{X_{1, \rm den}(z)} F(z),
\end{gather*}
where $X_{1, \rm den}(z)$ is given in equation~\eqref{eq:X1d} below, and $F(z)$ is a $k$-symmetric $p$-quasi periodic entire function of degree $2N-4$. Explicitly, we have
\begin{gather}\label{eq:F}
F(z)=C z \prod_{i=1}^{N-2} [z/\lambda_i, k/z\lambda_i],
\end{gather}
where $C, \lambda_1,\ldots, \lambda_{N-2}$ are some constants independent of $z$.
\end{Lemma}

\begin{proof} We put
\begin{gather*}
X_1(z)=\frac{D_1(z)}{\psi(z)}=\mu_1(z)P(z)Q(z/q)-P(z/q)Q(z).
\end{gather*}

(i) Obviously $X_1(z)$ is a $p$-periodic function. Due to the cancellations of the factors $[u_1\!/\!z][u_2\!/\!z]$, the denominator $X_{1, \rm den}(z)$ of $X_1(z)$ consists of $2(N+m+n)$ theta factors, hence $X_1(z)$ is of degree $2(N+m+n)$. We choose the normalization of $X_{1, \rm den}(z)$ as
\begin{gather}\label{eq:X1d}
X_{1, \rm den}(z)=\left(\prod_{i=1}^{N}[k/u_iz][u_{N+i}z/k]\right) [qu_1/z,qu_1 z/k]_{m}[qk/u_2z,qz/u_2]_{n},
\end{gather}
so that $X_{1, \rm den}(k/z)=\mu_1(z)^{-1}X_{1, \rm den}(z)$.

(ii) Due to the $k/q$-symmetry of $P(z)$, $Q(z)$ and equation~\eqref{eq:MK-rel}, we have
\begin{gather*}
X_1(k/z)=\mu_1(k/z)P(k/z)Q(k/qz)-P(k/qz)Q(k/z)\\
\hphantom{X_1(k/z)}{} =\mu_1(z)^{-1}P(z/q)Q(z)-P(z/q)Q(z)=-\mu_1(z)^{-1}X_1(z).
\end{gather*}
Combining this and $X_{1, \rm den}(k/z)=\mu_1(z)^{-1}X_{1, \rm den}(z)$, we see that the numerator $X_{1,\rm num}(z)$ is $k$-antisymmetric: $X_{1, \rm num}(k/z)=-X_{1, \rm num}(z)$.

(iii) By the Pad\'e interpolation condition, we have $D_1(q^{-s})=0$ for $s=0,1,\ldots, m+n-1$. Hence $X_{1, \rm num}(z)$ is divisible by $[z,k/z]_{m+n} \big[k/z^2\big]$.

From (i)--(iii), one obtain the desired result.\end{proof}

\begin{Lemma} \label{lemma:X23} We have
\begin{gather*}%\label{eq:XG}
{D_2(z)}=\psi(z)\frac{[z]_{m+n}[k/qz]_{m+n+1}}{X_{2, \rm den}(z)} G(z),\\
{D_3(z)}=\psi(z)\frac{[z]_{m+n+1}[k/qz]_{m+n}}{X_{3, \rm den}(z)} G(k/qz),
\end{gather*}
where $X_{2, \rm den}(z)$, $X_{3, \rm den}(z)$ is given in equation~\eqref{eq:X23d} below, $C'$ is a constant, and $G(z)$ is a~$p$-quasi periodic function of degree $N-1$ which can be written as
\begin{gather}\label{eq:G}
G(z)=C' z \prod_{i=1}^{N-1}[z/\xi_i], \qquad \prod_{i=1}^{N-1}\xi_i=\frac{k q^{n-1}}{\prod\limits_{i=1}^N u_i}.
\end{gather}
\end{Lemma}

\begin{proof} We put
\begin{gather*}
X_2(z)=\frac{D_2(z)}{\psi(z)}=\mu_2(z)P(z)\o{Q}(z/q)-\o{P}(z/q)Q(z),\\
X_3(x)=\frac{D_3(z)}{\psi(z)}=\mu_3(z)P(z)\o{Q}(z)-\o{P}(z)Q(z).
\end{gather*}

(i) Obviously $X_2(z)$, $X_3(z)$ are $p$-periodic elliptic functions. The denominators can be written as
\begin{gather}
 X_{2, \rm den}(z)=\prod_{i=N+1}^{2N}[u_iz/k] [qu_1/z,q u_1z/k]_m[q z/u_2, k/u_2 z]_n, \nonumber\\
 X_{3, \rm den}(z)=\prod_{i=N+1}^{2N}[u_i/q z] \big[u_1/z,q^2 u_1z/k\big]_m[q z/u_2, k/u_2 z]_n.\label{eq:X23d}
\end{gather}
Hence $X_2(z)$, $X_3(z)$ are both of degree $N+2m+2n$. We note that $X_{3, \rm den}(k/qz)=X_{2, \rm den}(z)$.

(ii) From $P(k/qz)=P(z)$, we have $\o{P}(k/qz)=\o{P}(z/q)$ and similarly we have $Q(k/qz)=Q(z)$, $\o{Q}(k/qz)=\o{Q}(z/q)$. Using these relations and equation~\eqref{eq:MK-rel}, we have
\begin{gather*}
X_3(k/qz)=\mu_3(k/qz)P(k/qz)\o{Q}(k/qz)-\o{P}(k/qz)Q(k/qz)\\
\hphantom{X_3(k/qz)}{} =\mu_2(z)P(z)\o{Q}(z/q)-\o{P}(z/q)Q(z)=X_2(z).
\end{gather*}

(iii) Due to the Pad\'e interpolation condition we have $D_2(q^{-s})=0$ for $s=0,\ldots, m+n-1$ and $D_3(q^{-s})=0$ for $s=0,\ldots, m+n$.

From (i)--(iii), we obtain the desired results.\end{proof}

\begin{Theorem}\label{th:cont}
By a suitable gauge transformation $y(z)=g(z) w(z)$, the $L_2$, $L_3$ equations take the following forms
\begin{gather}
L_2\colon \ F(z)\big[k/z^2\big]\o{y}(z/q)-G(z)A(k/z)y(z/q)+G(k/z)A(z) y(z)=0,\nonumber\\
L_3\colon \ \o{F}(z)\big[k/qz^2\big] y(z)-G(z)B(k/qz)\o{y}(z)+G(k/qz)B(z)\o{y}(z/q)=0,\label{eq:L2L3f}
\end{gather}
where $F(z)$, $G(z)$ are given by equations~\eqref{eq:F}, \eqref{eq:G}, and
\begin{gather*}
 A(z)=\prod_{i=1}^{N+1}[z/a_i], \qquad \{a_i\}_{i=1}^{N+1}=\big\{u_1q^m, u_2q^{-n}, u_3, \ldots, u_N,q\big\},\\
 B(z)=\prod_{i=1}^{N+1}[z/b_i], \qquad \{b_i\}_{i=1}^{N+1} = \big\{k/u_{N+1}, \ldots, k/u_{2N}, q^{-m-n}\big\}.
\end{gather*}
\end{Theorem}

\begin{proof} First, using Lemmas~\ref{lemma:X1} and~\ref{lemma:X23}, we rewrite the equations~(\ref{eq:L23D}) as
\begin{alignat*}{3}%\label{eq:L23inX}
& L_2\colon \ && F(z)\big[k/z^2\big]\o{w}(z/q)-G(z)[k/qz]\frac{[kq^n/u_2z]}{[k/u_2z]}\prod_{i=1}^N[k/u_iz]w(z/q) &\\
&&& {}+G(k/z)[z/q]\frac{[zq^n/u_2]}{[z/u_2]}\prod_{i=1}^N[z/u_i] w(z)=0,&\\
& L_3\colon \ && \prod_{i=1}^N \frac{[u_{N+i}/q]}{[k/q u_i]} \o{F}(z) \big[k/qz^2\big]w(z)
-G(z)\big[kq^{m+n-1}/z\big]\frac{[z u_1q^{m+1}z/k]}{[z u_1qz/k]}\prod_{i=N+1}^{2N}[u_i/qz]\o{w}(z) &\\
&&& {} +G(k/qz)\big[q^{m+n}z\big]q^m \frac{[z/u_1q^m]}{[z/u_1]}\prod_{i=N+1}^{2N}[zu_i/k]\o{w}(z/q)=0.&
\end{alignat*}
Then, by the gauge transformation $y(z)=[u_1/z,u_1 qz/k]_m w(z)$, we obtain
\begin{gather*}
L_2\colon \ q^{-m}F(z)\big[k/z^2\big]\o{y}(z/q)-G(z)A(k/z)y(z/q)+G(k/z)A(z) y(z)=0,\\
L_3\colon \ \prod_{i=1}^N \frac{[u_{N+i}/q]}{[k/q u_i]} \o{F}(z)\big[k/qz^2\big] y(z)-G(z)B(k/qz)\o{y}(z)+G(k/qz)B(z)\o{y}(z/q)=0.
\end{gather*}
The additional factors in front of $F(z)$, $\o{F}(z)$ can be absorbed into the normalization of $F(z)$, $\o{F}(z)$ by a $z$-independent gauge transformation of~$y(z)$. Hence, we arrive at the desired re\-sults~(\ref{eq:L2L3f}).
\end{proof}

\begin{Remark}\label{remark3.4}An explicit expression of the Pad\'e interpolants $P(z)$, $Q(z)$ is given by the determinant as follows
\begin{gather*}
P(z)=c \left|\begin{matrix}
{\mu}^P_{0,0}&\cdots&{\mu}^P_{0,m}\\
\vdots&\ddots&\vdots\\
{\mu}^P_{m-1,0}&\cdots&{\mu}^P_{m-1,m}\\
\psi_0(z)&\cdots&\psi_m(z)
\end{matrix}
\right|,\qquad
Q(z)=c \left|\begin{matrix}
{\mu}^Q_{0,0}&\cdots&{\mu}^Q_{0,n}\\
\vdots&\ddots&\vdots\\
{\mu}^Q_{n-1,0}&\cdots&{\mu}^Q_{n-1,n}\\
\phi_0(z)&\cdots&\phi_n(z)
\end{matrix}
\right|,
\end{gather*}
where $c$ is a constant and
\begin{gather*}
\psi_j(z)=\frac{[u_1,qu_1/k,k/u_3z,qz/u_3]_j}{[u_1/z,q u_1/z,k/u_3,q/u_3]_j},\qquad \phi_j(z)=\frac{[k/u_2,q/u_2,u_4/z,qu_4z/k]_j}{[k/u_2z,qz/u_2,u_4,qu_4/k]_j},\\
{\mu}^P_{i,j}={}_{2N+6}V_{2N+5}\left(\frac{k}{q},q^{-m-n},\frac{k}{u_1}q^{-j},\frac{k}{u_2}q^{m+n-i-1},\frac{k}{u_3}q^j,\frac{k}{u_4}q^i,\frac{k}{u_5},\dots ,\frac{k}{u_{2N}};q\right),\\
{\mu}^Q_{i,j}={}_{2N+6}V_{2N+5}\left(\frac{k}{q},q^{-m-n},{u_1}q^{m+n-i-1},{u_2}q^{-j},{u_3}q^i,{u_4}q^j,{u_5},\dots ,{u_{2N}};q\right),
\end{gather*}
and ${}_nV_{n-1}$ is the elliptic hypergeometric series \cite{Spi, Zhe} def\/ined by
\begin{gather}
{}_{n+5}V_{n+4}(a_0;a_1,\ldots,a_n;z)=\sum_{s=0}^{\infty}\frac{\big[a_0 q^{2s}\big]}{[a_0]}\prod_{i=0}^{n}\frac{[a_i]_s}{[q a_0/a_i]_s}z^s.
\end{gather}
The proof is completely the same as the case $N=3$ \cite{NTY13}. Application of the explicit formulae to the special solution of the isomonodromic systems will be considered elsewhere.
\end{Remark}

\section{Compatibility conditions}\label{section4}
In this section, we consider the equation~(\ref{eq:L2L3f}) forgetting about the connection with the interpolation problem. Namely, we restart with the following equations
\begin{alignat}{3}
& L_2\colon \ && F(z)\big[k/z^2\big]\o{y}(z/q)-G(z)A(k/z)y(z/q)+G(k/z)A(z) y(z)=0,& \nonumber\\
& L_3 \colon \ && \o{F}(z)\big[k/qz^2\big] y(z)-G(z)B(k/qz)\o{y}(z)+G(k/qz)B(z)\o{y}(z/q)=0, & \nonumber\\
&&& F(z)=C z \prod_{i=1}^{N-2} [z/\lambda_i, k/z\lambda_i],\qquad G(z)=z \prod_{i=1}^{N-1}[z/\xi_i], & \nonumber\\
 &&& A(z)=\prod_{i=1}^{N+1}[z/a_i],\qquad B(z)=\prod_{i=1}^{N+1}[z/b_i],& \label{eq:Lax23}
\end{alignat}
where $\{a_i, b_i\}_{i=1}^{N+1}$, $k$, $\ell$ are parameters and $C$, $\{\lambda_i\}_{i=1}^{N-2}$, $\{\xi_i\}_{i=1}^{N-1}$ are variables such as
\begin{gather*}
\o{a_i}=a_i, \qquad \o{b_i}=b_i, \qquad
\o{k}=k/q, \qquad \o{\ell}=q \ell,\qquad
k^2 \ell^2=q \prod_{i=1}^{N+1} a_ib_i, \qquad \prod_{i=1}^{N-1}\xi_i=\ell.
\end{gather*}

\begin{Proposition} As the necessary conditions for the compatibility, the pair of equations~$L_2$,~$L_3$ in~\eqref{eq:Lax23} gives the following equations for $\o{\lambda}_i$, $\o{C}$ and $\u{\xi_i}=T^{-1}(\xi_i)$. Namely
\begin{gather}\label{eq:Fup}
F(z)\o{F}(z)\big[k/z^2\big]\big[k/qz^2\big]={G(k/z)G(k/q z)U(z)},
\end{gather}
for $z=\xi_i$ $(1\leq i\leq N-1)$, and
\begin{gather}\label{eq:Gdn}
\frac{G(z)\u{G}(z)}{G(k/z)\u{G}(k/z)}=\frac{U(z)}{U(k/z)},
\end{gather}
for $z=\lambda_i$ $(1\leq i\leq N-2)$, where $U(z)=A(z)B(z)$.
\end{Proposition}

\begin{proof} When $z=\xi_i$, the terms in $L_2$, $L_3$ with coef\/f\/icient $G(z)$ vanishes, and we obtain equation~(\ref{eq:Fup}). Similarly, putting $z=\lambda_i$ in $L_2$ and
\begin{gather*}
\u{L_3}\colon \ F(z)\big[k/z^2\big] \u{y}(z)-\u{G}(z)B(k/z){y}(z)+\u{G}(k/z)B(z){y}(z/q)=0,
\end{gather*}
the terms with coef\/f\/icient $F(z)$ vanishes, and we have equation~(\ref{eq:Gdn}).
\end{proof}

The equations~(\ref{eq:Fup}), (\ref{eq:Gdn}) give the evolution equation for $2(N-2)$ variables $\{\lambda_i, \xi_i\}_{i=1}^{N-2}$.
In $N=3$ case, it can be written in a symmetric way as
\begin{gather*}
\frac{\xi_1^2}{\xi_2^2}\prod_{j=1}^2 \frac{[\xi_1/\lambda_j][\xi_1/\o{\lambda_j}]}{[\xi_2/\lambda_j][\xi_2/\o{\lambda_j}]}=\frac{U(\xi_1)}{U(\xi_2)}, \qquad
\frac{\lambda_1^2}{\lambda_2^2}\prod_{j=1}^2 \frac{[\lambda_1/\xi_j][\lambda_1/\u{\xi_j}]}{[\lambda_2/\xi_j][\lambda_2/\u{\xi_j}]}=\frac{U(\lambda_1)}{U(\lambda_2)},
\end{gather*}
where $\lambda_1\lambda_2=k$, $\xi_1\xi_2=\ell$. This is the elliptic Painlev\'e equation \cite{ORG, Sakai} in factorized form \cite{KNY15, NTY13}. Its Lax pair is obtained in \cite{Yamada09-2}, and the higher-order analogues are also given in~\cite{Rains11}.

\begin{Theorem}\label{th:44}
The equations~\eqref{eq:Fup}, \eqref{eq:Gdn} are sufficient for the compatibility of the Lax pair~\eqref{eq:Lax23}.$\!$
\end{Theorem}

\begin{proof}
Combining the equations $L_2$ and $L_3$ as the following diagrams:

\begin{center}\setlength{\unitlength}{1mm}
\begin{picture}(50,30)(-5,-5)
\put(-1,23){$\o{y}(\frac{z}{q})$}
\put(20,23){$\o{y}(z)$}
\put(-1,-3){$y(\frac{z}{q})$}
\put(20,-3){$y(z)$}
\put(41,-3){$y(qz)$}
\put(0,0){\line(1,0){20}}
\put(0,0){\line(0,1){20}}
\put(0,20){\line(1,-1){20}}
\put(1,21){\line(1,0){20}}
\put(21,1){\line(0,1){20}}
\put(1,21){\line(1,-1){20}}
\put(3,4){$L_2(z)$}
\put(10,14){$L_3(z)$}
\put(22,0){\line(1,0){20}}
\put(22,0){\line(0,1){20}}
\put(22,20){\line(1,-1){20}}
\put(24,4){$L_2(qz)$}
\put(-12,-3){$L_1$:}
\end{picture}
\qquad\qquad
\begin{picture}(50,30)(-5,-5)
\put(-1,23){$\o{y}(\frac{z}{q})$}
\put(20,23){$\o{y}(z)$}
\put(41,23){$\o{y}(qz)$}
\put(20,-3){$y(z)$}
\put(41,-3){$y(qz)$}
\put(23,21){\line(1,0){20}}
\put(43,1){\line(0,1){20}}
\put(23,21){\line(1,-1){20}}
\put(1,21){\line(1,0){20}}
\put(21,1){\line(0,1){20}}
\put(1,21){\line(1,-1){20}}
\put(32,14){$L_3(qz)$}
\put(10,14){$L_3(z)$}
\put(22,0){\line(1,0){20}}
\put(22,0){\line(0,1){20}}
\put(22,20){\line(1,-1){20}}
\put(24,4){$L_2(qz)$}
\put(-12,23){$\tilde{L}_1$:}
\end{picture}
\end{center}
we obtain the following three term relations for $y(z)$ or $\o{y}(z)$,
\begin{alignat}{3}
& L_1\colon \ && A(k/z) B(z)F(q z)\big[k/q^2 z^2\big]y (z/q ) -R(z)y(z) +A(q z)B(k/q z)F(z)\big[k/z^2\big] y(q z) =0,& \nonumber\\
& \tilde{L}_1\colon \ && A (k/qz) B(z)\o{F}(q z)\big[k/q^3 z^2\big]\o{y} (z/q) -\tilde{R}(z)\o{y}(z)&\nonumber\\
&&& {} +A(q z)B\big(k/q^2 z\big)\o{F}(z)\big[k/qz^2\big] \o{y}(q z) =0,&\label{eq:L1L1u}
\end{alignat}
where
\begin{gather}
R(z)=\frac{U(z) F(q z) G\big(\frac{k}{z}\big) \big[\frac{k}{q^2z^2}\big]}{G(z)}
+\frac{U\big(\frac{k}{q z}\big)F(z) G(q z) \big[\frac{k}{z^2}\big] }{G\big(\frac{k}{q z}\big)} \nonumber\\
\hphantom{R(z)=}{} -\frac{F(z)F(q z) \o{F}(z) \big[\frac{k}{z^2}\big]\big[\frac{k}{q z^2}\big]\big[\frac{k}{q^2z^2}\big]}{G(z) G\big(\frac{k}{q z}\big)},\nonumber\\
\tilde{R}(z)=\frac{U(qz) \o{F}(z) G\big(\frac{k}{q^2z}\big) \big[\frac{k}{q^2z^2}\big]}{G(qz)}
+\frac{U\big(\frac{k}{q z}\big)\o{F}(qz) G(z) \big[\frac{k}{q^3z^2}\big] }{G\big(\frac{k}{q z}\big)}\nonumber\\
\hphantom{\tilde{R}(z)=}{} -\frac{F(z) F(q z)\o{F}(z) \big[\frac{k}{z^2}\big]\big[\frac{k}{q z^2}\big]\big[\frac{k}{q^2z^2}\big]}{G(z) G\big(\frac{k}{q z}\big)}.\label{eq:R}
\end{gather}

Then the compatibility means the consistency
between triangle relations on the following $7$ points of the $3 \times 3$ grid:
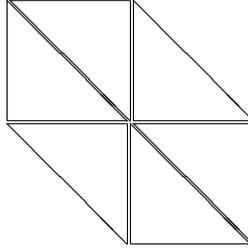
\begin{center}\setlength{\unitlength}{0.8mm}
\begin{picture}(50,30)(-5,-18)
\put(0,0.5){\line(1,0){20}}
\put(0,0.5){\line(0,1){20}}
\put(0,20.5){\line(1,-1){20}}
\put(0.5,20.5){\line(1,0){20}}
\put(20.5,0.5){\line(0,1){20}}
\put(0.5,20.5){\line(1,-1){20}}
\put(21,0.5){\line(1,0){20}}
\put(21,0.5){\line(0,1){20}}
\put(21,20.5){\line(1,-1){20}}
\put(0,0){\line(1,0){20}}
\put(20,-20){\line(0,1){20}}
\put(0,0){\line(1,-1){20}}
\put(20.5,-20){\line(1,0){20}}
\put(20.5,-20){\line(0,1){20}}
\put(20.5,0){\line(1,-1){20}}
\put(21,0){\line(1,0){20}}
\put(41,-20){\line(0,1){20}}
\put(21,0){\line(1,-1){20}}
\end{picture}
\end{center}
which is equivalently written as $\o{L_1}=\tilde{L}_1$. It is easy to check the condition $\o{L_1}=\tilde{L}_1$ for the coef\/f\/icients of $y(q z)$, $y(z/q)$ and $\o{y}(q z)$, $\o{y}(z/q)$. So the problem is to show $\o{R(z)}=\tilde{R}(z)$ under the equations (\ref{eq:Fup}), (\ref{eq:Gdn}).

For $L_1$ given in equation~(\ref{eq:L1L1u}), (\ref{eq:R}), one can check the following properties:
\begin{itemize}\itemsep=0pt
\item[(i)] $R(z)$ is holomorphic due to equation~(\ref{eq:Fup}), and it is a degree $4N+2$ theta function of base~$p$.
\item[(ii)] $R(k/qz)=-R(z)$, and hence $R(z)$ is divisible by $\big[k/qz^2\big]$.
\item[(iii)] The equation $L_1$ holds when
\begin{gather*}
\big\{\big[k/z^2\big]=0, \, y(z)=y(z/q)\big\} \quad {\rm or}\quad \big\{\big[k/q^2z^2\big]=0, \, y(z)=y(qz)\big\} \quad {\rm or}\\
\{F(z)=0, \, A(z)G(k/z)y(z)=A(k/z)G(z)y(z/q)\} \quad {\rm or}\\
\{F(qz)=0, \, A(qz)G(k/qz)y(qz)=A(k/qz)G(qz)y(z)\}.
\end{gather*}
\end{itemize}
Moreover, once the coef\/f\/icients of $y(q z)$, $y(z/q)$ in $L_1$ are f\/ixed as in equation~(\ref{eq:L1L1u}), the properties (i)--(iii) determine the coef\/f\/icient $R(z)$ uniquely.

Similarly, $\tilde{R}(z)$ in equation~(\ref{eq:L1L1u}) is characterized by the following conditions:
\begin{itemize}\itemsep=0pt
\item[(i)] $\tilde{R}(z)$ is a degree $4N+2$ theta function of base $p$.
\item[(ii)] $\tilde{R}\big(k/q^2z\big)=-\tilde{R}(z)$, and hence $\tilde{R}(z)$ is divisible by $\big[k/q^2z^2\big]$.
\item[(iii)] The equation $\tilde{L}_1$ holds when
\begin{gather*}
\big\{\big[k/qz^2\big]=0, \, \o{y}(z)=\o{y}(z/q)\big\} \quad {\rm or}\quad
\big\{\big[k/q^3z^2\big]=0, \, \o{y}(z)=\o{y}(qz)\big\} \quad {\rm or}\\
\big\{\o{F}(z)=0, \, A(z)\o{G}(k/qz)\o{y}(z)=A(k/qz)\o{G}(z)\o{y}(z/q)\big\} \quad {\rm or}\\
\big\{\o{F}(qz)=0, \, A(qz)\o{G}\big(k/q^2z\big)\o{y}(qz)=A(k/q^2z)\o{G}(qz)\o{y}(z)\big\},
\end{gather*}
where we used the equation~(\ref{eq:Gdn}) to rewrite the last two equations.
\end{itemize}
These characteristic properties show that $\o{R(z)}=\tilde{R}(z)$, hence
$\o{L_1}=\tilde{L}_1$ as desired.
\end{proof}

\subsection*{Acknowledgments}

The author is grateful to the organizers and participants of the lecture series at the university of Sydney (November~28--30, 2016) and the ESI workshop ``Elliptic Hypergeometric Functions in Combinatorics, Integrable Systems and Physics'' (Vienna, March 20--24, 2017) for their interests and discussions. He also thanks to referees for valuable comments and Dr.~H.~Nagao for discussions. This work is partially supported by JSPS KAKENHI (26287018).

\pdfbookmark[1]{References}{ref}
\LastPageEnding

\end{document}